\documentclass[sigconf,nonacm]{acmart}

\usepackage{algorithm}
\usepackage{acronym}
\usepackage{algpseudocode}
\usepackage{csquotes}
\usepackage{graphicx}
\usepackage{svg}
\usepackage{graphbox}
\usepackage{multicol}
\usepackage{amsmath}
\usepackage{mathtools}
\usepackage{pgfplots}
\usepackage{cancel}
\usepackage{physics}
\usepackage{xcolor,colortbl}
\usepackage{amsthm}

\newtheorem{lemma}{Lemma}
\newtheorem{proposition}{Proposition}
\newtheorem{case}{Case}
\newtheorem{therule}{Rule}

\newcommand{\SecPLF}{SecPLF}

\newcommand{\oldstuff}[1] {}

\usepackage{enumitem}
\setlist{  
  listparindent=\parindent,
  parsep=0pt,
}

\definecolor{Gray}{gray}{0.9}
\definecolor{Yellow}{rgb}{255,255,0}
\newcolumntype{y}{>{\columncolor{Yellow}}c}

\AtBeginDocument{%
  \providecommand\BibTeX{{%
    \normalfont B\kern-0.5em{\scshape i\kern-0.25em b}\kern-0.8em\TeX}}}

\setcopyright{acmcopyright}
\acmYear{2024}
\acmConference[ACM ASIACCS 2024]{The 19th ACM ASIA Conference on Computer and Communications Security}{1st to 5th of July, 2024}{Singapore}

\copyrightyear{2023}
\begin{document}
% \sloppy

\title{\SecPLF: Secure Protocols for Loanable Funds against Oracle Manipulation Attacks}

\author{Sanidhay Arora}
\affiliation{%
  \institution{University of Oregon}
  \city{Eugene}
  \country{USA}}
\email{sanidhay@uoregon.edu}

\author{Yingjiu Li}
\affiliation{%
  \institution{University of Oregon}
  \city{Eugene}
  \country{USA}}
\email{yingjiul@uoregon.edu}

\author{Yebo Feng}
\affiliation{%
  \institution{Nanyang Technological University}
  \country{Singapore}}
\email{yebo.feng@ntu.edu.sg}

\author{Jiahua Xu}
\affiliation{%
  \institution{University College London, The DLT Science Foundation}
  \city{London}
  \country{UK}}
\email{jiahua.xu@ucl.ac.uk}

\begin{abstract}
The evolving landscape of Decentralized Finance (DeFi) has raised critical security concerns, especially pertaining to Protocols for Loanable Funds (PLFs) and their dependency on price oracles, which are susceptible to manipulation.
The emergence of flash loans has further amplified these risks, enabling increasingly complex oracle manipulation attacks that can lead to significant financial losses. 
Responding to this threat, we first dissect the attack mechanism by formalizing the standard operational and adversary models for PLFs. Based on our analysis, we propose \SecPLF, a robust and practical solution designed to counteract oracle manipulation attacks efficiently. 
\SecPLF\ operates by tracking a price state for each crypto-asset, including the recent price and the timestamp of its last update. 
By imposing price constraints on the price oracle usage, \SecPLF\ ensures a PLF only engages a price oracle if the last recorded price falls within a defined threshold, thereby negating the profitability of potential attacks. 
Our evaluation based on historical market data confirms \SecPLF's efficacy in providing high-confidence prevention against arbitrage attacks that arise due to minor price differences.
\SecPLF\ delivers proactive protection against oracle manipulation attacks, offering ease of implementation, oracle-agnostic property, and resource and cost efficiency. 
\end{abstract}




\maketitle

\section{Introduction}

Decentralized Finance (DeFi) has dramatically reshaped the landscape of the financial sector in recent years, introducing a paradigm shift toward an inclusive and highly programmable system of finance~\cite{wood2014ethereum,werner2022sok,xu2022reap}. 
Rooted in the underlying technology, DeFi enables the execution of smart contracts that automate the delivery of financial services, rendering intermediaries unnecessary. 
While this democratizes access to financial instruments, it also introduces unique security risks and vulnerabilities. 
Attacks on DeFi platforms have become more and more frequent, sophisticated, and damaging~\cite{zhou2023sok}. 
Collectively, reported DeFi attacks have resulted in the loss of over \$3B in funds to date~\cite{rekt}, leading to substantial financial setbacks for investors. 
It is therefore of paramount importance to develop practical and effective strategies to defend against such attacks.

Among the myriad challenges that DeFi protocols face, a critical vulnerability in Protocols for Loanable Funds (PLFs) lies in their dependence on price oracles~\cite{zhou2023sok}. 
These external data sources, essential in furnishing market price data, have regrettably become a target for manipulation, thereby exposing associated PLFs to potential exploitation. 
A tampered oracle can trigger erroneous market price signals, inflicting severe repercussions across the DeFi landscape.
This issue is particularly pronounced with the advent of \emph{flash loans}, a novel smart contract functionality that can facilitate oracle manipulation attacks on PLFs.

\smallskip

Emerging in recent years within the DeFi arena, flash loans~\cite{wang2021towards} have empowered users with the capabilities to borrow significant capital amounts while incurring only gas fees. 
Given the fact that all entities---including flash loan providers, price oracles, and PLFs---operate on publicly accessible smart contracts, a malicious user could craft programs that leverage flash loans to manipulate a price oracle deftly. 
This manipulation sets the stage for considerable exploitation of any PLF relying on the compromised oracle~\cite{flash-loan}. 
Recent instances have seen several significant attacks on PLFs, with crypto-asset losses surpassing the \$100M mark~\cite{rekt}.

\smallskip

Regrettably, the prevention of oracle manipulation attacks proves to be challenging~\cite{zhou2023sok}, mainly as these attacks occur within a single transaction, inherently leaving little room for mitigation. 
Further complicating matters is the intricate interaction between multiple smart contracts across various DeFi platforms, amplifying the complexity of financial transactions and inhibiting the discovery of security vulnerabilities. 
Although some approaches have been proposed (e.g., multiple oracle sources~\cite{synthetix}, price feeds and averaging~\cite{adams2022uniswap}, timelock mechanism~\cite{ezzat2022timelock}, staking and reputation system~\cite{pauwels2022zkkyc}), oracle manipulation attacks have not been fully addressed. 
On the other hand, similar price manipulation problems have been studied and addressed in centralized finance (CeFi).
For instance, traditional finance uses circuit breakers to pause trading amid significant price fluctuations or manipulations to protect investors~\cite{Santoni1993CircuitBA}. However, this approach has not been properly explored in DeFi.

In this paper, we seek to counteract oracle manipulation attacks on PLFs, particularly those facilitated by flash loans. 
We propose \SecPLF, a robust approach specifically designed to provide an effective and cost-efficient solution. 
This approach, which PLFs can easily implement, serves as a powerful tool to combat and avert these potentially catastrophic attacks, thereby enhancing the overall security of the DeFi landscape.

The key idea of \SecPLF\ lies in tracking a price state for each crypto-asset. 
This state comprises of \textit{(i)} the spot price of a crypto-asset, and \textit{(ii)} the timestamp of the block in which this price state was last updated. 
\SecPLF\ imposes specific constraints on this price state and the usage of the price oracle. 
These constraints are designed to ensure that the PLF only engages a price oracle if the last recorded price of the asset is within a defined threshold. 
We calibrate this threshold to be equivalent to the minimum price distortion an adversary needs to attain a profit from the attack transaction.
This strategy effectively deters attackers, as it renders the attack unprofitable, thereby ensuring they are discouraged from initiating it in the first place.
In contrast with existing solutions, \SecPLF\ presents the following contributions:
\begin{enumerate}[leftmargin=*]
\item \textbf{Proactive Safeguarding:} \SecPLF\ provides robust and effective prevention against oracle manipulation attacks. Consequently, attackers are unable to successfully launch an attack initially, causing negligible impacts on the DeFi protocol operations.
\item \textbf{Re-configurable and Steerable:} With two hyper-parameters, $z$ and $\epsilon$, \SecPLF\ enables a PLF to quantify its security in terms of arbitrage risk ($z$) and under-collateralization risk ($\epsilon$). The flexibility of these parameters provides the PLF with configuration control over their protocol based on these risks.
\item \textbf{Ease of Implementation:} \SecPLF\ only employs fundamental smart contract functionalities for implementation, requiring minimal modifications to adapt to a variety of DeFi systems.
\item \textbf{Resource Efficiency:} The computational resources required for the \SecPLF\ algorithm are minor compared to the overall expenses of operating a PLF.
\item \textbf{Oracle Agnostic:} \SecPLF\ functions as an integral security measure in the architecture of a standard PLF, irrespective of the oracle source. This oracle-agnostic property allows \SecPLF\ to integrate into any standard PLF, enhancing its applicability.
\end{enumerate}

Our analysis and evaluation have demonstrated the aforementioned advantages. 
Drawing upon market data from the last three years, we have found that \SecPLF\ can achieve high levels of confidence (e.g., $1 - 10^{-5}$) in fending off potential arbitrage opportunities, provided the parameters are appropriately configured.
Furthermore, \SecPLF\ stands out as a resource and cost-efficient solution. 
Even when we consider a worst-case scenario overhead costs, these costs remain considerably minor when compared to the operational costs of a PLF. 
Thus, \SecPLF\ offers practical, robust, and economical protection against flash loan-driven oracle manipulation attacks.

The remainder of this paper is structured as follows: Section~\ref{sec: related} provides a survey of related work, while Section~\ref{sec: pre} introduces the necessary background information required for understanding our methodology. We then formalize the operational model for a standard PLF in Section~\ref{sec: PLF} and describe the adversary model in Section~\ref{sec: ad-model}. 
Section~\ref{sec: secPLF} details our proposed solution, \SecPLF. We evaluate our approach in Section~\ref{sec: analysis} and draw conclusions in Section~\ref{sec:concl}.
\section{Related Work}\label{sec: related}

This section presents the related work, including DeFi security, oracle manipulation attacks, and traditional finance security.

\subsection{Decentralized Finance Security}

DeFi security, a burgeoning research area, addresses the challenges posed by the rapid growth and adoption of DeFi platforms and the increasing complexity of financial attacks targeting them~\cite{li2022security,zhou2023sok, wang2021blockeye}. 
The ``blockchain oracle problem'' is a pivotal concern, leading to significant losses in DeFi projects~\cite{caldarelli2020understanding,caldarelli2021blockchain}.
Solutions like decentralized oracles and security platforms are suggested but not widely implemented~\cite{oracle-problem}. \SecPLF\ offers an innovative solution for preventing oracle manipulation attacks in this context.

Oracle manipulation attacks in PLFs are of particular concern, causing substantial damages~\cite{rekt}. 
Massimo et al. highlighted the security issues in PLFs, including smart contract vulnerabilities and oracle manipulation attacks~\cite{bartoletti2021sok}. \SecPLF\ targets these flash loan-driven attacks.

\subsection{Price Oracles and Manipulation Attacks}

There is an increasing focus on oracle manipulation attacks, especially involving flash loans~\cite{angeris2020improved, flash-loan, flashot, first-step}. Early research by Yixin et al. demonstrated the potential for rapid exploitation through atomic transactions, though without proposing countermeasures~\cite{flash-loan}. Various strategies have emerged to protect PLFs, focusing on the integrity of oracle data:

\begin{itemize}[leftmargin=*]
\item \textbf{Price Feeds and Averaging}: Strategies like time-weighted average price (TWAP) introduced by Uniswap V3 aim to mitigate outliers~\cite{adams2022uniswap}, but Makinga et al. show some manipulative activities can bypass TWAP oracles~\cite{mackinga2022twap}.
\item \textbf{Timelock Mechanisms}: Delaying oracle usage can mitigate attacks, but may not always be effective~\cite{ezzat2022timelock}.
\item \textbf{Staking and Reputation Systems}: These systems incentivize accuracy but are susceptible to manipulative behaviors~\cite{pauwels2022zkkyc}.
\item \textbf{Multiple Oracle Sources}: Using a weighted average from various sources increases accuracy but adds complexity~\cite{synthetix, hackmd-oracle}.
\end{itemize}

\subsection{Traditional Finance Security}

In traditional finance, circuit breakers prevent significant price manipulations~\cite{qin2021cefi}. 
Studies affirm their effectiveness in combating short-term manipulations~\cite{Santoni1993CircuitBA, Lauterbach1993StockMC, li2021impacts}. 
However, despite their efficacy, the adaptation of circuit breakers to the DeFi environment has not yet been achieved.
\SecPLF\ adapts a similar concept for the DeFi space to counter oracle manipulation attacks.

\section{Background}\label{sec: pre}

This section provides essential background on blockchains and DeFi applications to facilitate a better understanding of \SecPLF.

\subsection{Decentralized Finance (DeFi)}

DeFi is a blockchain-based form of finance that doesn't rely on central financial intermediaries such as brokerages, exchanges, or banks to offer financial services~\cite{werner2022sok}.
Instead, it utilizes smart contracts~\cite{Grish2018} on blockchains, predominantly Ethereum. 
Offering services like lending, insurance, and trading, DeFi's growth has been exponential, leading to considerable challenges, particularly in security and privacy~\cite{defi-market}, necessitating continual research.

\subsection{Oracles} 

An oracle serves as a bridge between the blockchain environment and the real world~\cite{sok-oracle,bowen2020oracle}. 
As blockchain systems are designed to be isolated from external influences to ensure security, most blockchain applications, particularly in DeFi projects, require specific information from outside the blockchain to trigger the execution of smart contracts. 
Oracles fulfill this need, operating as on-chain APIs that bring external information into the smart contracts. 
They can report on a wide array of data, such as the exchange rate of ETH/USD on Binance or the winners of the 2021 NBA Championship. 
Moreover, oracles can be bi-directional, not only fetching data but also transmitting information out to the real world.

\subsection{Flash Loans}

Flash loans, a novel DeFi tool, have transformed the ways of obtaining loans in the blockchain environment. 
Unlike traditional loans, flash loans allow users to borrow any amount of assets without collateral, provided that the borrowed assets are returned within the same transaction~\cite{chandler2022defi,first-step}. 
This unique feature enables a myriad of use cases, including arbitrage, collateral swapping, and self-liquidation, among others.
However, it also introduces new types of risks and attacks, like the notorious oracle manipulation attacks that could lead to significant financial losses~\cite{rekt}.

\subsection{Automated Market Maker-based Decentralized Exchanges (AMM DEXs)}

AMM DEXs represent a novel approach to decentralized trading, which eliminates the need for an order book~\cite{sok-amm}. 
Instead of matching buyers and sellers to determine prices and execute trades, AMM DEXs use a mathematical formula to set the price of a token.
AMMs allow digital assets to be traded in a permissionless and automated way by using liquidity pools rather than a traditional market of buyers and sellers. 
Popular examples include Uniswap~\cite{uniswap}, Curve~\cite{curve}, Balancer~\cite{balancer}, Bancor~\cite{bancor}, TerraSwap\cite{terraswap}, and Raydium~\cite{raydium}.

\subsection{Protocols for Loanable Funds (PLFs)}

PLFs are a critical component of the DeFi ecosystem, facilitating the majority of lending and borrowing activities within blockchain networks~\cite{gudgeon2020defi}. 
Through automated and programmatically-enforced smart contracts, PLFs allow users to lend their crypto assets in a pool from which borrowers can borrow, often requiring over-collateralization ~\cite{perez2021}.
PLFs dynamically adjust interest rates based on supply and demand, aiming to balance the liquidity in the lending pools. 
Some notable examples of PLFs include Aave, Compound, and MakerDAO. 
While these protocols have democratized access to financial services, they are not without risk, as seen in several instances of security breaches and manipulative attacks.

\section{Standard PLF Model}\label{sec: PLF}

This section formalizes a standard model that PLFs utilize in practice. 
First, we introduce the risk considerations and define the security notion of \emph{safe collateralization} that a standard PLF uses to mitigate these risks. 
Next, we describe the factors and mechanisms a PLF employs and define a collateralization-safe PLF model based on this security notion. 
Finally, we highlight the benefits and potential risks of the usage of price oracles for this PLF model.

\subsection{Standard Risks in PLFs}\label{subsec: plf-security}

The two primary risks that lead to insolvency and concern a standard PLF\textemdash highlighted by Kao et al. in \cite{Kao2020AnAO}\textemdash are as follows.
\begin{enumerate}[leftmargin=*]
    \item Arbitrage risk: It refers to the potential of traders and participants to exploit the price rate discrepancies between PLFs, creating profit opportunities.
    To avoid these price discrepancies, PLFs use price oracles to fetch the latest market price for each asset. 
    These discrepancies typically arise for a short time based on the frequency of price oracle usage.
    \item Under-collateralization risk: 
    It occurs when the collateral value held by a borrower is insufficient to cover the value of the borrowed assets, including interest accrued.
    If a borrower becomes under-collateralized, the PLF may not be able to fully recover the loans, leading to potential losses for lenders and liquidity providers. 
    It threatens the stability and solvency of the PLF.
\end{enumerate}

\subsection{Safe Collateralization Approach}
In the next two sub-sections, we model a $PLF$ that addresses the risk of under-collateralization using a \emph{safe collateralization} approach.
The definition of Safe collateralization is as follows.

\begin{definition}[Safe collateralization]
\label{def:safe-col}
It refers to an over-collateralization approach that aims to ensure that the protocol remains solvent. 
Specifically, this approach ensures that the total value of outstanding loans and obligations does not exceed the available funds and assets.\footnote{This definition is inspired from \cite{bartoletti2021sok}.}
\end{definition}

This model involves the calculation of a liquidation threshold to secure the PLF from under-collateralization. 
The formal definition of the liquidation threshold is as follows. 

\begin{definition}[Liquidation threshold]
\label{def:liq-thresh}
    It is the specified value at which an asset must be sold or liquidated to limit additional losses or manage risk. 
    When a collateral value drops below this threshold, PLF liquidates the collateral to mitigate losses.
\end{definition}

A PLF employing the safe collateralization approach assumes that loans are secured by collateral.
Therefore, liquidations are incentivized to recover loans if borrowers fail to repay. 
However, collateral liquidation is exposed to risks.
The incentive to liquidate is only effective if the liquidated value of seized collateral is higher than the value of the repaid loan amount, implying a profit. 
Smart contracts are used to automate the liquidation process which reduces the risk of fraud or manipulation.
We now describe the factors and mechanisms that PLFs use to determine these thresholds in case of liquidations.

\subsubsection{Liquidation threshold calculation factors}

Intuitively, the liquidation threshold is a balance between minimizing the risk of default and maximizing the protocols' profitability. 
The liquidation threshold for a PLF is typically determined by two parameters, which are:

\begin{itemize}[leftmargin=*]
    \item Value of a collateral asset ($c$): It refers to the value of a collateralized asset held by a borrower. 
    If the protocol has a low value for $c$, it may indicate that the PLF is close to the risk of under-collateralization. 
    \item Value of an outstanding loan ($l$): It refers to the value of the collateralized loan. 
    If the protocol has a high value for $l$, it may indicate the risk of under-collateralization.
\end{itemize}

\subsubsection{Liquidation threshold calculation mechanism}

\begin{table}[t]
    \centering
    \caption{PLF: State parameter notations}    
    \label{tab:liquid}
    \begin{tabular}{c|p{6cm}}
    \hline
         \textbf{Notation} & \textbf{Meaning} \\
         \hline
         $c$ & Value of a collateral asset held by user\\
         $l$ & Value of an outstanding loan \\
         $E$ & Collateralization ratio; $E = \frac{c}{l}$ \\
         $\epsilon$ & Safe collateralization ratio\\
         $LT$ & Liquidation threshold; $LT = \epsilon \cdot l$ \\
         \hline
    \end{tabular}
\end{table}

Table~\ref{tab:liquid} provides a brief description of the PLF parameters relevant to the calculation of the liquidation threshold ($LT$). 
The key idea of the algorithm to calculate $LT$ in a PLF using ``Safe collateralization" (Definition~\ref{def:safe-col}) is described below. 
A safe collateralization ratio $\epsilon$ is established for a PLF based on the volatility of the collateral asset and the desired level of risk \cite{aave-docs}.
$LT$ is then set at a level that ensures that the collateralization ratio ($E$) does not fall below the safe collateralization ratio ($\epsilon$), i.e. $E \ge \epsilon$. 
In the case of a standard PLF using a safe collateralization approach, the liquidation threshold for a loan $l$ can be defined as the safe collateralization ratio ($\epsilon$) times the value of the loan ($l$) i.e. $LT = \epsilon \cdot l$.

\subsection{Collateralization-safe PLF Model}

To model $PLF$ using safe collateralization, we use the parameters defined in Table~\ref{tab:liquid} to define the safe-collateralization rule i.e. $E \ge \epsilon$. 
The safe-collateralization rule ensures that the total value of the collateral is greater than or equal to the safe collateralization ratio times the value of the total loan.
We now define collateralization-safe PLF which employs the security notions provided in Definitions ~\ref{def:safe-col}, and ~\ref{def:liq-thresh}.
\footnote{This definition is inspired from~\cite{Kao2020AnAO, bartoletti2021sok}.}
\begin{definition}[Collateralization-safe PLF]\label{def:e-safe}
    A $PLF$ is collateralization-safe if the safe-collateralization rule is satisfied, i.e.
    \begin{therule}\label{rule:e-safe}
    $E \ge \epsilon \iff c \ge LT \iff c \ge \epsilon \cdot l.$
    \end{therule}
\end{definition}

The safe-collateralization rule can be used to monitor the state of PLF and ensure that it remains solvent and secure from under-collateralization.
If the rule is violated, it indicates that the protocol is at risk of under-collateralization, and corrective measures need to be taken. 
If any collateral $c$ falls below their liquidation threshold $LT$ i.e. $c < \epsilon \cdot l$, PLF liquidates $c$ in time so that no loss is incurred \cite{aave-docs, compound-docs}.
The main focus of this security model is to ensure that the PLF remains solvent and follows the security notions of Definitions~\ref{def:safe-col}, ~\ref{def:liq-thresh}, and ~\ref{def:e-safe}. 
This model allows for quantification of the degree of security against under-collateralization using the safe collateralization ratio $\epsilon$.

\subsection{Price Oracle Usage in PLFs}

Several PLFs employing the safe collateralization approach use price oracles to fetch the spot price of crypto-assets.
PLFs use these oracles to monitor safe collateralization at least every 10 hours \cite{aave-docs, compound-docs}. Later, we use this minimum frequency of price oracle usage (of 10 hours) in our practicality analysis (Section~\ref{sec: practical}).
Additionally, PLFs use these price oracles to keep up with the latest market price, reducing the risk of arbitrage attacks caused due to price shocks.

However, the usage of price oracles may introduce the risk of arbitrage from a novel type of attack observed in recent years, called an oracle manipulation attack. 
Since price oracles are critical in the calculation of the liquidation threshold, and henceforth the borrowing limit on any loans, it can lead to under-collateralization if this price is manipulated. 
In addition, a novel DeFi service called flash loans can help facilitate oracle manipulation attacks with essentially no risk to the attacker. 
Moving forward, we consider this collateralization-safe PLF (from Definition~\ref{def:e-safe}) using price oracles as the model for a standard PLF.
We refer to it as $PLF$ (italic) for simplicity. 

Further, we consider this $PLF$ to always satisfy Rule~\ref{rule:e-safe}. 
This rule ensures that in case $PLF$ uses an oracle price that contradicts Rule~\ref{rule:e-safe}, that is if the price of a collateral drops below the liquidation threshold, it liquidates the collateral immediately. 
Hence, this standard $PLF$ is secure from under-collateralization according to Definition~\ref{def:e-safe}. 
Next, in Section~\ref{sec: ad-model}, we present an adversary model for this $PLF$ formalizing the aforementioned flash loan-driven oracle manipulation attack.

\section{Adversary Model}\label{sec: ad-model}

In this section, we first define a basic model of blockchains and the adversary. 
Next, we define relevant notations and formalize an oracle manipulation attack on a standard $PLF$ defined in Section~\ref{sec: PLF}.
We then model a general flash loan-driven oracle manipulation attack transaction $\mathcal{T}$, provide its formal attack steps, and define an oracle-manipulation-secure $PLF$.
Moreover, we illustrate an example of a formal oracle manipulation attack transaction $\mathcal{T}$ in Figure~\ref{fig:attack}.

\subsection{Basic Blockchain and Adversary Model}

\begin{table}[t]
    \centering
    \caption{Blockchain and adversary model notations}\label{tab:notations}
    \begin{tabular}{c|c}
    \hline
        \textbf{Notation} & \textbf{Meaning}\\
        \hline
        $\mathsf{S}_i$ & $i^{th}$ state of the blockchain\\
        $\mathcal{T}_i$ & $i^{th}$ sub-transaction; $\mathcal{T}_i: \mathsf{S}_{i-1} \rightarrow \mathsf{S}_i$\\
        $\mathcal{T}$ & Attack transaction: $(\mathcal{T}_1, \cdots, \mathcal{T}_n)$\\
        $\tau_i$ & $i^{th}$ transaction-step\\
        \hline
    \end{tabular}
\end{table}

We model the interaction between users and the blockchain as a state transition system.
We assume one computationally bounded and economically rational adversary $\mathbb{A}$. 
The adversary is not required to provide its collateral to perform the attacks described below.
We assume that the adversary is financially capable of paying any transaction fees associated with the network.
Now, consider $\mathbb{A}$ creates a set $M$ of $m$ malicious smart contracts.
It is reasonable to assume that the adversary $\mathbb{A}$ and all $m$ malicious smart contracts can interact with each other according to $\mathbb{A}$'s wish.

Here we define the notations of the adversary model summarized in Table~\ref{tab:notations}. 
Let the initial state of the blockchain be denoted by $\mathsf{S}_0$. 
To execute an attack, $\mathbb{A}$ initiates a transaction $\mathcal{T}$ by calling a malicious smart contract belonging to $M$.
A \emph{transaction} is defined as an indexed-sequence of $n$ atomic state-transition-functions, i.e. $\mathcal{T} = (\mathcal{T}_1, \cdots, \mathcal{T}_n)$. \emph{Atomic state transitions} are indivisible and irreducible series of state-change operations on the blockchain such that either all occurs, or nothing occurs.
Here, $\mathcal{T}_i: \mathsf{S}_{i-1} \rightarrow \mathsf{S}_i$ represents an atomic state transition function, i.e. it is a definite series of operations on the blockchain state. 
Let $\mathcal{T}_i$ be referred to as $i^{th}$ \emph{sub-transaction}.
The state of the blockchain after transaction $\mathcal{T}$ is $\mathsf{S}_n$, i.e. $\mathcal{T}: \mathsf{S}_0 \rightarrow \mathsf{S}_n$.

Now consider an equivalent indexed-sequence of the transaction as  $\mathcal{T} = (\tau_1, \tau_2, \cdots, \tau_k)$.
Here, $\tau_i$ is an indexed sequence of sub-transactions, i.e. $\tau_i = (\mathcal{T}_x, \cdots, \mathcal{T}_y)$ where $0 < x < y \leq n$ and $x, y \in \mathbb{Z}^+$.
$\mathsf{S}_{y}$ denotes the state of the blockchain after the transaction step $\tau_i$, i.e. $\tau_i: \mathsf{S}_{x - 1} \rightarrow \mathsf{S}_y$ where $i < x < y \leq n$.
Trivially, this sequence must be mutually exclusive, topologically sorted, and completely exhaustive for a given transaction $\mathcal{T}$.

\subsection{Formalizing Oracle Manipulation Attack}

$PLF$s using price oracles may be vulnerable to flash loan-driven oracle manipulation attacks. 
In this attack, the adversary can manipulate the price oracle of a collateralized asset by using flash loans. 
This manipulation is done by executing a large trade on a $DEX$ that provides the oracle that $PLF$ uses.
The attacker can use the flash loan to borrow large amounts of cryptocurrency assets. 
This borrowed asset may be used to deposit on the PLF as collateral for loans, which is also the cost of this attack.
The adversary then uses the majority of the flash loan amount to execute trades on a DEX that relies on the same oracle used by $PLF$.
By executing a large enough trade on $DEX$, the adversary can manipulate the price of the collateral, causing it to appear more valuable than it is. 
Once the price of the collateral has been manipulated, the attacker can use it to secure a loan on the PLF, borrowing a larger amount of funds. 
The adversary profits an amount equal to this loan minus the cost of the attack, as her profit.

\subsubsection{Formal attack transaction steps}\label{subsec:formal-steps}

\begin{table}[t]
    \centering
    \caption{Transaction-steps ($\tau_i$)}\label{tab:tx-steps}
    \begin{tabular}{c|c}
        \hline
        \textbf{Notation} & \textbf{Meaning}\\
        \hline
        $\mathcal{F}(X, \mathbf{A})$ & Take $X$ amount of $\mathbf{A}$ as flash loan \\
        $\mathcal{D}(\mathbf{A}, \mathbf{cA})$ & Deposit $\mathbf{A}$ as collateral to get $\mathbf{cA}$ \\ 
        $\mathcal{S}(\mathbf{A}, \mathbf{B})$ & Swap: Sell $\mathbf{A}$ to buy $\mathbf{B}$\\
        $\mathcal{B}(\mathbf{cA}, \mathbf{B})$ & Borrow $\mathbf{B}$ using $\mathbf{cA}$ on $PLF$ \\
        $\mathcal{P}(X, \textbf{A})$ & Payback the flash loan with $X$ amount of $\mathbf{A}$\\
        \hline
    \end{tabular}
\end{table}

Here, we model a formal general flash loan-driven oracle manipulation attack transaction $\mathcal{T}$. 
Table~\ref{tab:tx-steps} summarizes a short description of the crucial transaction steps that are necessary for an oracle manipulation attack.
Consider a PLF using a price oracle $\mathbb{O}_{\mathbf{A}}$ provided by a Decentralized Exchange $DEX$ to determine the price of an asset $\mathbf{A}$.
The adversary $\mathbb{A}$ initiates a transaction $\mathcal{T}$, an indexed sequence of $k$ transaction steps i.e. $\mathcal{T} = (\tau_1, \cdots, \tau_k)$. 
The critical transaction steps, along with their reasoning, that are necessary for the attack are described below in topological order:

\begin{itemize}[leftmargin=2em]
    \item[$\mathcal{F}$] $(X, \mathbf{A});$ Take $X$ amount of asset $\mathbf{A}$ in flash loan as $\tau_1$.
    \item[$\mathcal{S}$] $(\mathbf{A}, \mathbf{B});$ This transaction-step must use a small fraction of $X$. 
    \item[$\mathcal{D}$] $({\mathbf{B}}, \mathbf{cB});$ $\mathbb{A}$ deposits $Y$ amount of asset $\mathbf{B}$ on $PLF$, which is the cost of this attack. 
    Here `$\mathbf{c}$' in $\mathbf{cA}$ denotes a \emph{collateralized-asset}. 
    $\mathbb{A}$ plans to distort $\mathbb{O}_\mathbf{B}$ by distorting the price of the pair $\mathbf{A}/\mathbf{B}$ on $DEX$, which is used in the price oracle $\mathbb{O}_{\mathbf{B}}$.
    \item[$\mathcal{S}$] $(\mathbf{A}, \mathbf{B});$ $\mathbb{A}$ distorts the price of $\mathbb{O}_\mathbf{B}$ by a factor of $\Theta$ using the remaining amount of $\mathbf{A}$. 
    $PLF$ receives $\mathbb{O}_{\mathbf{B}}$ as input which represents the price of $\mathbf{cB}$ (linearly proportional). 
    Note that the majority of the flash loan amount goes here.
    \item[$\mathcal{B}$] $(\mathbf{cB}, \mathbf{C});$ $\mathbb{A}$ redeems the collateral $\mathbf{cB}$ at the distorted price to borrow a loan $L$ in asset $\mathbf{C}$. 
    Here, $\mathbb{A}$ can borrow this loan up to her maximum borrowing limit from Definition~\ref{def:e-safe}.
    \item[$\mathcal{S}$] $(\mathbf{B}, \mathbf{A});$ This step is reverting the swap in which the adversary manipulated the price to get back $X$ amount of $\mathbf{A}$ to pay back the flash loan.
    \item[$\mathcal{P}$] $(X, \mathbf{A});$ Payback the flash loan in asset $\mathbf{A}$ as $\tau_k$. 
    The adversary gains a profit $G$ which is calculated as the borrowed loan $L$ minus the cost of the attack.
\end{itemize}

\subsubsection{Formal attack transaction notations}

\begin{table}[t]
    \centering
    \caption{Formal attack transaction notations}\label{tab:T-notations}
    \begin{tabular}{c|c}
        \hline
        \textbf{Notation} & \textbf{Meaning}\\
        \hline
        $\mathbb{O}_\mathbf{B}$ & Oracle price of asset $\mathbf{B}$ before the transaction $\mathcal{T}$\\
        $Y$ & Amount of $\mathbf{B}$ deposited in $\mathcal{D}(\mathbf{B}, \mathbf{cB})$\\
        $\Theta$ & Distortion factor of $\mathbb{O}_\mathbf{B}$ right before $\mathcal{B}(\mathbf{cB}, \mathbf{C})$\\
        $\max(L)$ & Borrowing limit of $\mathbb{A}$ (in USD) right before $\mathcal{B}(\mathbf{cB}, \mathbf{C})$\\
        $L$ & Amount of loan in $\mathbf{C}$ borrowed in step $\mathcal{B}(\mathbf{C}, \mathbf{cB})$\\
        $G$ & Profit (or gain) of $\mathbb{A}$ from the transaction $\mathcal{T}$\\
        \hline
    \end{tabular}
\end{table}

\begin{figure*}[t]
    \centering
    \includegraphics[width=0.63\textwidth]{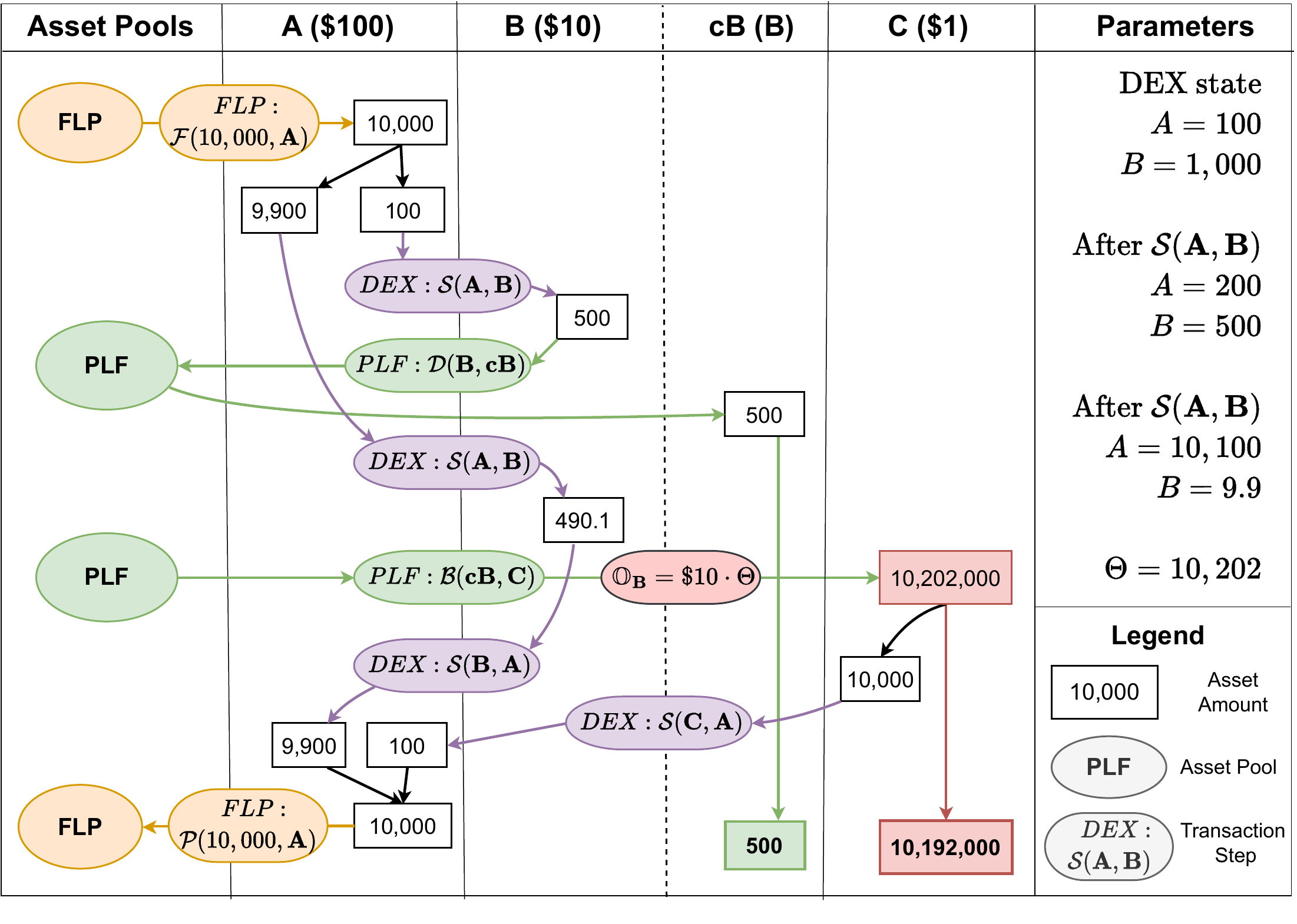}
    \caption{Flash loan-driven oracle manipulation attack example of transaction $\mathcal{T}$: $\mathbb{A}$ distorts the initial price of \textbf{B} ($\$10$) by a factor of $\Theta = 10,202$. $\mathbb{A}$ profits ($G$) $10,192,000$ amount of \textbf{C} priced at $1$ USD (highlighted in red).}\label{fig:attack}
\end{figure*}

Table~\ref{tab:T-notations} summarizes the notations used in the aforementioned formal attack transaction $\mathcal{T}$. 
Let $\mathbb{O}_\mathbf{B}$ be the initial oracle price of asset $\mathbf{B}$ before the transaction $\mathcal{T}$. 
In transaction step $\mathcal{S}(\mathbf{A}, \mathbf{B})$, $\mathbb{A}$ distorts this initial price by a factor of $\Theta$. 
This distortion increases $\mathbb{A}$'s borrowing limit $\max(L)$ by the same factor.
It is calculated from Definition~\ref{def:e-safe} as the USD value of the deposited collateral $Y\ \mathbf{B}$ after a distortion of $\Theta$ in $\mathbb{O}_\mathbf{B}$ over the safe collateralization ratio $\epsilon$, i.e. 
\begin{equation}\label{eq-max-L}
\max(L) = \frac{\mathbb{O}_\mathbf{B} \cdot Y \cdot \Theta}{\epsilon}.
\end{equation}
To gain a profit from the attack transaction $\mathcal{T}$, $\mathbb{A}$ borrows a loan $L$ in asset $\mathbf{C}$. 
Here $L$ calculates using $\max(L)$ (from Equation~\ref{eq-max-L}) as 
\begin{equation}\label{eq-L}
L = \frac{\max(L)}{\mathbb{O}_\mathbf{C}}.
\end{equation}

$\mathbb{A}$'s profit (or gain) can be calculated in USD as her borrowing limit minus the cost of the attack.
Let $G$ be the profit (in USD) of $\mathbb{A}$ in $\mathcal{T}$.
Hence, $G$ calculates as $\max(L)$ (from Equation~\ref{eq-max-L}) minus the value of the deposited collateral $Y$ (which is the cost of the attack), i.e.
\begin{equation}\label{eq-G}
G = \mathbb{O}_\mathbf{B} \cdot Y \cdot \left(\frac{\Theta}{\epsilon} - 1\right).
\end{equation}

We now calculate an upper bound for $\mathbb{A}$'s profit $G$ as $\max(G)$ (from Equation~\ref{eq-G}) using the maximum values of $Y$ and $\Theta$ as $\max(Y)$ and $\max(\Theta)$, respectively. 
Here, $\max(G)$ calculates as
\begin{equation}\label{eq-max-g}
\max(G) = \mathbb{O}_\mathbf{B} \cdot \max(Y) \cdot \left(\frac{\max(\Theta)}{\epsilon} - 1\right).
\end{equation}
Note that, since $\mathbb{A}$ uses a negligible fraction of the flash loan to get the collateral $Y\ \mathbf{B}$, $\max(G)$ highly depends on $\max(\Theta)$.

\subsubsection{Oracle-manipulation-secure PLF}

Here we define an Oracle-manipulation-secure PLF based on the adversary model, the formal attack transaction $\mathcal{T}$, and the upper bound on $\mathbb{A}$'s profit $\max(G)$.

\begin{definition}[Oracle-manipulation-secure PLF]\label{def: sec-plf}
    A $PLF$ is secure against flash loan-driven oracle manipulation attacks if the attack transaction $\mathcal{T}$ is unsuccessful. 
    Equivalently, if $\mathbb{A}$'s maximum profit $\max(G)$ (from Equation~\ref{eq-max-g}) in the transaction $\mathcal{T}$ is $0$, i.e.
    \begin{therule}\label{rule-maxTheta-maxG}
    $\mathbb{O}_\mathbf{B} \cdot \max(Y) \cdot \left(\frac{\max(\Theta)}{\epsilon} - 1\right) = 0 \iff \max(\Theta) = \epsilon.$
    \end{therule}
\end{definition}

Intuitively, Definition~\ref{def: sec-plf} states that a PLF is secure against the attack transaction $\mathcal{T}$ if the maximum price distortion factor ($\max(\Theta)$) is equal to the safe collateralization ratio ($\epsilon$) i.e. (Rule~\ref{rule-maxTheta-maxG}) $$\max(\Theta) = \epsilon \iff \max(G) = 0.$$

\subsection{Oracle Manipulation Attack Example}

\subsubsection{Attack example illustration} 

Figure~\ref{fig:attack} illustrates a flash loan-driven oracle manipulation attack on a PLF.
We utilize the Flashot diagram provided by Yixin et al.~\cite{flashot} to illustrate the attack transaction $\mathcal{T}$ as described in Section~\ref{subsec:formal-steps}. 
Below, we describe the terms and notations along with some key ideas of the attack example in Figure~\ref{fig:attack}.
The initial market price of assets $\mathbf{A}, \mathbf{B}$, and $\mathbf{C}$ are $\$100, \$10$, and $\$1$ before the transaction $\mathcal{T}$, respectively (mentioned after the assets). 
The timeline of the transaction $\mathcal{T}$ is shown from top to bottom.
Asset pools like Flash Loan Provider ($FLP$), and $PLF$ are on the left. 
The parameters shown on the right are the number of liquidity reserves of asset pool $\mathbf{A}/\mathbf{B}$ on the Decentralized Exchange ($DEX$) which is initially set as $A = 100$ and $B = 1,000$. 
The values of the parameters are calculated by $DEX$ using its AMM algorithm, which operates using a conservation function. 
In this example, we use Uniswap's conservation function: $A \times B = K$, where $K$ is a constant ($K= 100,000$). 

\par $PLF$ receives the price of asset $\mathbf{B}$ from the price oracle provided by $DEX$ as its input, $\mathbb{O}_\mathbf{B}$. $\mathbb{A}$ wants to manipulate the price of asset $\mathbf{cB}$, which is proportional to $\mathbb{O}_\mathbf{B}$. 
After this price is manipulated, $\mathbb{A}$ can then borrow any asset on $PLF$ up to its borrowing limit. 
This borrowing limit is directly proportional to \textit{(i)} the price of $\mathbf{cB}$, which was manipulated; and \textit{(ii)} inversely proportional to the safe collateralization ratio $\epsilon$ (from Definition~\ref{def:e-safe}).
Since, $\mathbb{A}$ can profit an amount up to this borrowing limit, we take $\epsilon = 500 \%$ as it is significantly more than the typical upper bound used for this ratio \cite{aave-docs}. 
This case represents the worst-case scenario for the adversary as the more the safe collateralization ratio, the less the borrowing limit, implying less profit for the adversary $\mathbb{A}$.

\subsubsection{Transaction steps in  $\mathcal{T}$}

The transaction steps ($\tau_i \in \mathcal{T}$) of the attack transaction $\mathcal{T} = (\tau_1, \cdots, \tau_k)$ shown in Figure~\ref{fig:attack} are explained as follows.
\begin{itemize}[leftmargin=2em]
    \item[$\mathcal{F}$] $(X, \mathbf{A})$; $\mathbb{A}$ borrows a flash loan of $X = 10,000\ \mathbf{A}$ in $\tau_1$. \item[$\mathcal{S}$] $(\mathbf{A},\mathbf{B})$; $\mathbb{A}$ swaps $100\ \mathbf{A}$ on $DEX$ providing the price oracle to get $500\ \mathbf{B}$. 
    \item[$\mathcal{D}$] $(\mathbf{B}, \mathbf{cB})$; $\mathbb{A}$ deposits $Y = 500\ \mathbf{B}$ on $PLF$ to get $500\ \mathbf{cB}$. 
    \item[$\mathcal{S}$] $(\mathbf{A}, \mathbf{B})$; $\mathbb{A}$ uses the remaining flash loan amount of $9,900\ \mathbf{A}$  to get $490.1\ \mathbf{B}$. $\mathbb{A}$ distorts the price of asset $\mathbf{B}$ on $DEX$, which is proportional to the reserve amount of $\mathbf{A}$ over the reserve amount of $\mathbf{B}$. The distortion in $\mathbb{O}_\mathbf{B}$ calculates as $$\Theta = \frac{10,100}{9.9} \cdot \frac{1,000}{100} \approx 10,202.$$ 
    \item[$\mathcal{B}$] $(\mathbf{cB}, \mathbf{C})$; $\mathbb{A}$ borrows a loan $L$ against the collateral $\mathbf{cB}$ to gain a profit $G$ from the attack transaction $\mathcal{T}$ i.e $G > 0$. $PLF$ receives the distorted oracle price $\mathbb{O}_\mathbf{B} = \$10 \cdot \Theta$ as input (highlighted in red). 
    The borrowing limit of $L$ is calculated in USD (from Equation~\ref{eq-max-L}) as the value of $500\ \mathbf{cB}$ at the distorted price over the safe collateralization ratio $\epsilon$, i.e. $$\max(L) = \frac{\$ 10 \cdot 500 \cdot 10,202}{500\ \%} = \$10,202,000.$$  $\mathbb{A}$ borrows a loan $L$ up to this borrowing limit in an asset $\mathbf{C}$ (from Equation~\ref{eq-L}), i.e. $$L = \frac{\max(L)}{\mathbb{O}_\mathbf{C}} = 10,202,000 \ \mathbf{C}.$$
    \item[$\mathcal{S}$] $(\mathbf{B}, \mathbf{A})$; $\mathbb{A}$ swaps back $490.1\ \mathbf{B}$ on $DEX$ to get $9,900\ \mathbf{A}$.
    \item[$\mathcal{S}$] $(\mathbf{C}, \mathbf{A})$; $\mathbb{A}$ uses $10,000\ \mathbf{C}$ to get $100\ \mathbf{A}$. This amount is the cost of the attack, which is equivalent to the value of deposited collateral $500\ \mathbf{B}$ i.e. $100 \cdot \mathbb{O}_\mathbf{A}$ (in USD).
    \item[$\mathcal{P}$] $(X, \mathbf{A})$; Finally, $\mathbb{A}$ pays back the flash loan in $\tau_k$. $\mathbb{A}$'s profit $G$ (highlighted in red) is calculated from Equation~\ref{eq-G} as $$G = \$10 \cdot 500 \cdot \left( \frac{10,202}{500 \%} - 1 \right) = \$ 10,192,000.$$
\end{itemize}

The collateral provided to $PLF$, that is $500\ \mathbf{cB}$, remains there (highlighted in green). 
This amount is the cost of the attack, which is negligible compared to $\mathbb{A}$'s profit. 
Hence, it can be ignored and left as locked collateral in $PLF$. 
Moreover, flash loan providers might require a small fraction of the fee to use their service.
Typically this fee is less than $0.1\%$ of the flash loan borrowed.
Since this fraction is negligible to the price distortion and $\mathbb{A}$' return on investment, we neglect this fee for simplicity.

\section{\SecPLF: Oracle Manipulation Secure PLF Solution}\label{sec: secPLF}

\begin{algorithm*}[t!]
\caption{\SecPLF: Secure Price Oracle Algorithm for asset $\mathbf{A}$}\label{algo1}
    \begin{algorithmic}[1]
    \Require{$\mathbb{O}_\mathbf{A}$, $\mathbb{B}$, $\epsilon$} \Comment{Oracle price of asset $\mathbf{A}$, block in which the $PLF$ receives the oracle, and safe collateralization ratio, respectively.}
    \Ensure{$P_\mathbf{A}$} \Comment{The price of asset $\mathbf{A}$ that the $PLF$ uses as a safe price against oracle manipulation attacks.}
    \If{$\mathbb{B}.id > \mathsf{A}.id$}
        \Comment{Update the state $\mathsf{A}$ if and only if the input is received in a new block.}
        \State{$\mathsf{A} \gets (\mathbb{B}.id, \min(\mathbb{O}_\mathbf{A}, \mathsf{A}.p \cdot \epsilon))$}
        \Comment{Maximum discrepency in price after $(\mathbb{B}.id - \mathsf{A}.id)$ block(s) is $\max(0, \mathbb{O}_\mathbf{A} - \mathsf{A}.p \cdot \epsilon)$.}
    \EndIf
    \State{$P_\mathbf{A} \gets \min(\mathbb{O}_\mathbf{A}, \mathsf{A}.p)$} \Comment{Update the output with the minimum of stored price $\mathsf{A}.p$ and oracle price $\mathbb{O}_\mathbf{A}$.}
    \State{\Return $P_\mathbf{A}$}
    \end{algorithmic}
\end{algorithm*}

In this section, we present \SecPLF\ which is an algorithm designed to secure a standard $PLF$ from flash loan-driven oracle manipulation attacks according to Definition~\ref{def: sec-plf}. 
Given an input price oracle $\mathbb{O}_\mathbf{A}$, \SecPLF\ algorithm outputs a price $P_\mathbf{A}$, which is based on a safe threshold to prevent the attack transaction $\mathcal{T}$. 
We first present an overview of \SecPLF\ explaining the key ideas behind this algorithm.
Next, we proceed by presenting the \SecPLF\ algorithm.
Then, we provide a Theorem stating that \SecPLF\ is an oracle-manipulation-secure PLF according to Definition~\ref{def: sec-plf} and show that \SecPLF\ prevents the attack transaction illustrated in Figure~\ref{fig:attack}. 
Further, we quantify the arbitrage risk resulting from the price differentials caused due to the imposed constraints on the output price.
Later, in Section~\ref{sec: analysis}, we quantify this risk based on this difference and use it in practicality analysis of \SecPLF.

\subsection{Overview}

\begin{table}[t]
    \caption{\SecPLF\ Algorithm Notations}
    \label{tab:design-notations}
    \begin{tabular}{c|p{6cm}}
    \hline
        \textbf{Notation} & \textbf{Meaning} \\
        \hline
        $\mathbb{O}_{\mathbf{A}}$ & Price of asset $\mathbf{A}$ provided by the oracle (input)\\
        $P_\mathbf{A}$ & The price of asset $\mathbf{A}$ that \SecPLF\ uses (output)\\
        $\mathbb{B}$ & Block in which the PLF receives  $\mathbb{O}_\mathbf{A}$ as input\\
        $\mathsf{A}$ & State of asset $\mathbf{A}$: $(id, p)$; the block index in which it was last updated, and price of asset $\mathbf{A}$\\
        \hline
    \end{tabular}
\end{table}

Following is a brief overview of the algorithm. 
The notations used in \SecPLF\ are summarized in Table~\ref{tab:design-notations}. 
This algorithm takes in three inputs \textit{(i)} $\mathbb{O}_\mathbf{A}$, which is the price oracle of an asset $\mathbf{A}$ that $PLF$ uses; \textit{(ii)} $\mathbb{B}$, the block in which $PLF$ receives the oracle; and \textit{(iii)} $\epsilon$, the safe collateralization ratio.
It stores a state $\mathsf{A}$ for each asset $\mathbf{A}$, in case $\mathbb{O}_\mathbf{A}$ is not safe to use. 
Finally, it gives an output $P_\mathbf{A}$, which denotes a safe price for asset $\mathbf{A}$ that prevents oracle manipulation attacks on $PLF$. 

The algorithm is explained as follows.
Initialize a state $\mathsf{A} = (id, p)$ in the PLF smart contract state for each asset $\mathbf{A}$ on $PLF$.
Here, $p$ is the last stored price of asset $\mathbf{A}$, and $id$ is the block index in which $\mathsf{A}$ was last updated. 
After fetching the price oracle $\mathbb{O}_\mathbf{A}$, update $\mathsf{A}$ if and only if the block index of the price oracle is greater than the last updated $id$ in $\mathsf{A}$.
This condition ensures that the state $\mathsf{A}$ updates at most once during the execution of each attack transaction $\mathcal{T}$. 
Update this state with the minimum of the oracle price $\mathbb{O}_\mathbf{A}$ and $\epsilon$ times $\mathsf{A}.p$.
Then, use the minimum of Oracle price $\mathbb{O}_\mathbf{A}$ and the last updated price state $\mathsf{A}.p$ as the output $P_\mathbf{A}$ of the algorithm. 
This restriction on the Oracle price ensures that $\mathbb{A}$ is unable to gain any profits from the transaction $\mathcal{T}$, which prevents oracle manipulation attacks (according to Definition~\ref{def: sec-plf}). 

However, using a stored price with the imposed constraints instead of the latest oracle price might allow some arbitrage attack opportunities on $PLF$. We define the Price-discrepancy case in Section~\ref{sec:quantify-risk}. It is the only case in the \SecPLF\ algorithm that allows arbitrage opportunities.
Then, we introduce a parameter $\Delta^\mathbb{B}_\mathbf{A}$ to quantify this risk.
Further, in Section~\ref{sec: analysis}, we show that this arbitrage case occurs with minimal probability in the practical setting.

\subsection{\SecPLF\ Algorithm}

The \SecPLF\ algorithm, given in Algorithm~\ref{algo1}, prevents oracle manipulation attacks according to Definition~\ref{def: sec-plf}. 
The algorithm is summarized as follows:

\begin{enumerate}[leftmargin=*]
    \item Let $\mathbb{O}_\mathbf{A}$ denote the oracle price of asset $\mathbf{A}$ as the input. 
    And, let $P_\mathbf{A}$ denote the price of asset $\mathbf{A}$ that $PLF$ uses to prevent oracle manipulation attacks, as the output of the algorithm.
    \item Store a state $\mathsf{A} = (id, p)$ for each asset $\mathbf{A}$.
    $\mathsf{A}.p$ stores the price of $\mathbf{A}$, and $\mathsf{A}.id$ stores the index of the block in which $\mathsf{A}$ was last updated.
    \item Let $\mathbb{B}.id$ denote the index of the current block in which $PLF$ receives the oracle $\mathbb{O}_\mathbf{A}$ as input. 
    Update the state $\mathsf{A}$ if and only if the block index $\mathbb{B}.id$ is greater than the stored state index $\mathsf{A}.id$, i.e.
    \begin{therule}\label{rule-update-A}
    $\mathsf{A} \gets (\mathbb{B}.id, \min(\mathbb{O}_\mathbf{A}, \mathsf{A}.p \cdot \epsilon)) \iff \mathbb{B}.id > \mathsf{A}.id.$
    \end{therule}
    \item Note that, in \SecPLF, the safe collateralization ratio $\epsilon$ is also the maximum factor of change allowed in $\mathsf{A}.p$ from its last updated state (from Rule~\ref{rule-update-A}).
    \item Now, update $P_\mathbf{A}$ as the minimum of oracle price $\mathbb{O}_\mathbf{A}$ and last updated price state $\mathsf{A}.p$, i.e. \begin{therule}\label{rule-update-P}
    $P_\mathbf{A} \gets \min(\mathbb{O}_\mathbf{A}, \mathsf{A}.p).$
    \end{therule}
    \item Finally, return $P_\mathbf{A}$ as the output, denoting the safe-to-use price of asset $\mathbf{A}$ against oracle manipulation attacks according to Definition~\ref{def: sec-plf}.
\end{enumerate}

\subsection{\SecPLF\ Theorem}

In this subsection, we provide a theorem to show that \SecPLF\ is an Oracle-manipulation-secure PLF according to Definition~\ref{def: sec-plf}. 
We start by stating a proposition on the \SecPLF\ algorithm below.

\begin{proposition}\label{prop1}
The \SecPLF\ algorithm ensures that the maximum achievable distortion $\max(\Theta)$ in the attack transaction $\mathcal{T}$ is equal to the safe collateralization ratio $\epsilon$, i.e. $\max(\Theta) = \epsilon$ (Rule~\ref{rule-maxTheta-maxG}).
\end{proposition}

To prove this proposition, we first define the following lemma and consequently prove it.

\begin{lemma}\label{lemma1}
    The \SecPLF\ algorithm ensures that the state $\mathsf{A}$ updates at most once during the execution of each attack transaction $\mathcal{T}$.
    Specifically, the state $\mathsf{A}$ updates at most once for all the $n$ sub-transactions in $\mathcal{T} = (\mathcal{T}_1, \cdots, \mathcal{T}_n)$, i.e.
    $$\mathsf{A} \gets (\mathbb{B}.id, \cdot) \iff \exists! i \in [1..n]\ \forall_{i = 1}^{n} \mathcal{T}_i \in \mathcal{T}.$$
\end{lemma}

\begin{proof}
The price state $\mathsf{A}$ only updates if the index of the current block is greater than the index of the block where it last updated, i.e. (Rule~\ref{rule-update-A}) $$\mathsf{A} \gets (\mathbb{B}.id, \cdot) \iff \mathbb{B}.id > \mathsf{A}.id.$$
Once it updates, i.e. $\mathsf{A}.id \gets \mathbb{B}.id$, it can no longer be updated in the same transaction $\mathcal{T}$. 
To prove this, consider a transaction $\mathcal{T}$ consisting of $n$ sub-transactions, i.e $\mathcal{T} = (\mathcal{T}_1, \cdots, \mathcal{T}_n)$. 
Now trivially, all sub-transactions in the same transaction must share the same block index $\mathbb{B}.id$, i.e.
\begin{equation}\label{eq-T-B}
    \mathcal{T}_i.\mathsf{id} = \mathcal{T}_j.\mathsf{id} = \mathbb{B}.id\ \forall i,j \in [1..n],
\end{equation}
where $\mathcal{T}_i.\mathsf{id}$ denotes the index of the block stored in the sub-transaction $\mathcal{T}_i \in \mathcal{T}$.
After the first update $\mathsf{A}.id = \mathbb{B}.id$, and since $\mathsf{A}$ only updates if $\mathbb{B}.id > \mathsf{A}.id$ (from Rule~\ref{rule-update-A})
Equation~\ref{eq-T-B} and Rule~\ref{rule-update-A} implies the aforementioned rule, i.e.
\begin{therule}\label{rule-update-A-once}
$\mathsf{A} \gets (\mathbb{B}.id, \cdot) \iff \exists! i \in [1..n]\ \forall_{i = 1}^{n} \mathcal{T}_i \in \mathcal{T}.$
\end{therule}

Therefore, Rule~\ref{rule-update-A-once} ensures that $\mathsf{A}$ updates at most once during the execution of each attack transaction $\mathcal{T}$.
Hence proved.
\end{proof}

The proof of Proposition~\ref{prop1} is as follows.

\begin{proof}
Given $\mathcal{T} = (\tau_1,\cdots,\tau_k)$, Lemma~\ref{lemma1} shows that any update in $\mathsf{B}$ happens at most once during the execution of the attack transaction $\mathcal{T}$.
After the first update, the maximum value of $\mathsf{B}.p$ is restricted by a factor of $\epsilon$ from its last value (from Rule~\ref{rule-update-A}), i.e. 
\begin{therule}\label{rule-Bp}
$\mathsf{B}.p \gets \min(\mathbb{O}_\mathbf{B}, \mathsf{B}.p \cdot \epsilon) \implies \max(\mathsf{B}.p) \gets \mathsf{B}.p \cdot \epsilon.$        
\end{therule}
Rule~\ref{rule-Bp} and Lemma~\ref{lemma1} show that \SecPLF\ ensures that $\mathsf{B}.p$ cannot be distorted by more than a factor of $\epsilon$ in transaction $\mathcal{T}$. 

Since $P_\mathbf{B} \gets \min(\mathbb{O}_\mathbf{B}, \mathsf{B}.p)$ (from Rule~\ref{rule-update-P}), $\epsilon$ is equivalent to the maximum distortion factor $\max(\Theta)$ of $\mathbb{O}_\mathbf{B}$.
This rule is implied in \SecPLF\ as it uses the price $P_\mathbf{B}$ instead of the oracle price $\mathbb{O}_\mathbf{B}$ as the output of the algorithm, i.e. 
\begin{therule}\label{rule-maxg-0}
$P_\mathbf{B} \gets \min(\mathbb{O}_\mathbf{B}, \mathsf{B}.p) \implies \max(\Theta) = \epsilon.$
\end{therule}
Therefore, Rule~\ref{rule-maxg-0} ensures that the maximum achievable distortion $\max(\Theta)$ in the attack transaction $\mathcal{T}$ is equal to the safe collateralization ratio $\epsilon$, i.e. $\max(\Theta) = \epsilon$. 
Hence proved.
\end{proof}

The theorem is stated as follows.

\begin{theorem}\label{theorem1}
A standard $PLF$ employing the \SecPLF\ algorithm, using $\epsilon$ as safe collateralization ratio and $\mathbb{O}_\mathbf{A}$ as a price oracle for asset $\mathbf{A}$, is Oracle-manipulation-secure according to Definition~\ref{def: sec-plf}.
\end{theorem}

\begin{proof}
Proposition~\ref{prop1} states that in \SecPLF, $\max(\Theta) = \epsilon$. 
This rule implies that $\mathbb{A}$'s maximum profit $\max(G)$ from the transaction $\mathcal{T}$ is $0$, i.e. (Rule~\ref{rule-maxTheta-maxG} from Definition~\ref{def: sec-plf})  $$\max(\Theta) = \epsilon \iff \max(G) = 0.$$

Since the adversary cannot gain any profits, she would be unable to pay back the flash loan, which would result in an unsuccessful transaction.
Hence proved.
\end{proof}

\subsection{Example of Attack Prevention}

Here we show that \SecPLF\ algorithm prevents the attack transaction $\mathcal{T}$ demonstrated in the oracle manipulation attack example in Figure~\ref{fig:attack}.
Notice that \SecPLF\ algorithm takes affect in the transaction step $\mathcal{B}(\mathbf{cB}, \mathbf{C})$.
Specifically, when $PLF$ fetches the distorted price oracle $\mathbb{O}_\mathbf{B}= \$10 \cdot \Theta$ (highlighted in red).
In this example, the distortion factor $\Theta = 10,202$. Here, $PLF$ fetches the distorted price oracle $\mathbb{O}_\mathbf{B}$ as its input.
The output of the \SecPLF\ algorithm, that is $P_\mathbf{B}$, denotes the safe price of asset $\mathbf{B}$ that $PLF$ must use to prevent the transaction $\mathcal{T}$. 

For the attack transaction $\mathcal{T} = (\tau_1, \cdots, \tau_k)$ in Figure~\ref{fig:attack}, the maximum value of $P_\mathbf{B}$ is restricted by the upper bound $\mathsf{B}.p \cdot \epsilon$ (from Rule~\ref{rule-update-A} and~\ref{rule-update-P} in the \SecPLF\ algorithm). 
Because the distortion created by $\mathbb{A}$, that is $\Theta$, is more than the maximum allowed distortion in $\SecPLF$, that is $\epsilon$, $PLF$ uses the output price as $\mathsf{B}.p \cdot \epsilon$  i.e. $$\Theta > \epsilon \implies P_\mathbf{B} \gets \mathsf{B}.p \cdot \epsilon,$$ where $\Theta = 10,202$ and $\epsilon = 5$.

The constraint $\max(\Theta) = \epsilon$ in Rule~\ref{rule-maxTheta-maxG}, and the equivalence of $\max(\Theta)$ and $\epsilon$ in Rule~\ref{rule-maxg-0} implies that $\mathbb{A}$ is unable to gain a profit, i.e. $$\max(\Theta) = \epsilon \iff \max(G) = 0.$$ 
Therefore, \SecPLF\ prevents the transaction $\mathcal{T}$ as $\mathbb{A}$ would fail to execute the final transaction-step $\tau_k = \mathcal{P}(10,000, \mathbf{A})$ to pay back the flash loan.

\subsection{Quantifying the Arbitrage Risk of \SecPLF}\label{sec:quantify-risk}

The primary reason $PLF$s use price oracles is to mitigate the risk of potential arbitrage attacks caused due to price shocks in the market. 
Although Theorem~\ref{theorem1} proves that \SecPLF\ is an Oracle-manipulation-secure PLF according to Definition~\ref{def: sec-plf}, it may allow for rare arbitrage opportunities in case of price shocks. 
The constraints imposed on the oracle price $\mathbb{O}_\mathbf{A}$ may introduce price differences in the output price $P_\mathbf{A}$, creating arbitrage opportunities. 

To quantify this arbitrage risk, we first define the \emph{Price-discrepancy case} in \SecPLF\ as the only case in which arbitrage risk is possible.
Specifically, there is an arbitrage opportunity in \SecPLF\ if and only if the algorithm receives an input price oracle $\mathbb{O}_\mathbf{A}$ such that $\mathbb{O}_\mathbf{A} > \mathsf{A}.p$ in Rule~\ref{rule-update-P}.
We refer to this case as the Price-discrepancy case, i.e. 
\begin{case}[Price-discrepancy]\label{case-z}
$\mathbb{O}_\mathbf{A} > \mathsf{A}.p.$
\end{case}

To quantify the risk in the Price-discrepancy case, given the inputs $\mathbb{O}_\mathbf{A}$, $\mathbb{B}$, and $\epsilon$,
% $T \in \mathbb{Z}^*$ blocks ago. $\mathsf{A}.p$ 
let $\Delta^\mathbb{B}_\mathbf{A}$ denote the price difference between the input oracle price $\mathbb{O}_\mathbf{A}$ and the output price $P_\mathbf{A}$ in block $\mathbb{B}$. 
Since $\mathbb{O}_\mathbf{A} > \mathsf{A}.p$, the output price $P_\mathbf{A}$ gets the value $\mathsf{A}.p$ (from Rule~\ref{rule-update-P}).
Therefore, $\Delta^\mathbb{B}_\mathbf{A}$ calculates as \begin{equation}\label{eq-delta-B}
\Delta^\mathbb{B}_\mathbf{A} = \mathbb{O}_\mathbf{A} - \mathsf{A}.p.
\end{equation}
Here, $\Delta^\mathbb{B}_\mathbf{A}$ denotes the price difference created in block $\mathbb{B}$ due to the constraints imposed on $\mathbb{O}_\mathbf{A}$.
Therefore, the parameter $\Delta^\mathbb{B}_\mathbf{A}$ quantifies the potential arbitrage risk in the Price-discrepancy case. 

Further, to analyze the arbitrage risk using $\Delta^\mathbb{B}_\mathbf{A}$, we must use the maximum possible value of this parameter in each block $\mathbb{B}$, which calculates as 
\begin{equation}\label{eq-max-delta}
    \max(\Delta^\mathbb{B}_\mathbf{A}) = \max(\mathbb{O}_\mathbf{A} - \mathsf{A}.p).
\end{equation} 

Next, in Section~\ref{sec: analysis}, we use this parameter in our empirical analysis to show that the Price-discrepancy case (Case~\ref{case-z}) does not occur in the real-world setting with high confidence of $z = 1 - 10^{-5}$. We refer to this case as the Confidence case, i.e.
\begin{case}[Confidence]\label{case-confidence}
$\mathbb{P}(\max(\Delta^\mathbb{B}_\mathbf{A}) \le 0) \ge z.$
\end{case}

In frequent statistics, confidence value ($z$) refers to the probability that a population parameter ($\max(\Delta^\mathbb{B}_\mathbf{A})$) will fall within a specific range a certain proportion of times. 
We calculate this probability value using the Cumulative Density Function (CDF) of $\max(\Delta^\mathbb{B}_\mathbf{A})$. 

\section{\SecPLF: Analysis}\label{sec: analysis}

In this section, we present the practical analysis of \SecPLF. 
We illustrate that the arbitrage opportunities generated due to the potential price differences caused by the \SecPLF\ algorithm do not occur with high confidence.
We do this by providing an empirical analysis of the market data collected over the past three years.
Next, we analyze the practicality of the \SecPLF\ algorithm to justify our contributions.
Finally, we provide a comparative study among the existing solutions introduced in Section~\ref{sec: related} and \SecPLF.

\subsection{Price-discrepancy Case Risk Analysis}\label{sub-sec: risk-analysis}

\begin{table}[t]
    \centering
    \caption{Price-discrepancy case (Case~\ref{case-z}) analysis notations} \label{tab: arbitrage-risk-notations}
    \begin{tabular}{c|p{6.5cm}}
        \hline
        \textbf{Notation} & \textbf{Meaning} \\
        \hline
        $D_\mathbf{A}$ & $D_\mathbf{A} = (d_1, \cdots, d_{N})$; $N$ minutes of market price data of $\mathbf{A}$ since September 1, 2020, where $N \approx 1.57 \cdot 10^6$\\ 
        $z$ & Confidence value $z = 1 - 10^{-5}$; the probability of the non-occurrence of Price-discrepancy case\\
        $T$ & Maximum number of minutes after which $PLF$ receives $\mathbb{O}_\mathbf{A}$ as input for each asset $\mathbf{A}$ \\
        $T_z$ & Maximum number of minutes $T$ such that the Confidence case (Case~\ref{case-confidence}) is satisfied\\
        $\max(\Delta^M_\mathbf{A})$ & $\mathbf{A}$'s maximum possible price difference at the $M^{th}$ minute in the Price-discrepancy case\\
        \hline
    \end{tabular}
\end{table}

Table~\ref{tab: arbitrage-risk-notations} summarizes the notations used in the Price-discrepancy case risk analysis. 
To demonstrate that the arbitrage risk due to price shocks is minimal in the real-world setting, we aim to illustrate that based on the past three years of market data, the Price-discrepancy case does not occur with a high confidence of $z$, i.e. the Confidence case is satisfied.
Since we set $z = 1 - 10^{-5}$, it ensures that the probability of non-occurrence of the Price-discrepancy case is $z$, which is sufficiently high for any practical purposes.

\subsubsection{Data collection}
We collect the market price data of 15 crypto-assets from the Crpyto-compare API \cite{crypto-compare}.
We choose a market capitalization value ranging from \$90 Million (\textbf{\textit{PERP}}) to \$750 Billion (\textbf{\textit{BTC}}). 
As of September 2023, this range accounts for over $99\%$ of the total cryptocurrency market capitalization value. 
Specifically, we store this data as $D_\mathbf{A}$, a per-minute time-series of $N$ data points for each asset $\mathbf{A}$ i.e. $D_\mathbf{A} = (d_1, \cdots, d_{N})$. Here, $d_i \in D_\mathbf{A}$ stores the close-price data of $\mathbf{A}$ in the $i^{th}$ minute starting from September 1st, 2020.
We collect this data from September 1st, 2020, to September 1st, 2023, which amounts to approximately 1.57 million data points for each crypto-asset $\mathbf{A}$, or a total of approximately 23.55 million data points i.e. $N \approx 1.57 \cdot 10^6$.

To analyze the arbitrage risk, we must analyze the value of the parameter $\max(\Delta^\mathbb{B}_\mathbf{A})$ (from Equation~\ref{eq-max-delta}) for each data set $D_\mathbf{A}$.
To calculate the value of $\max(\Delta^\mathbb{B}_\mathbf{A})$, where $\mathbb{B}.id = M$, we can replace $\mathbb{O}_\mathbf{A}$ with the market price of $\mathbf{A}$ in that block. 
However, since we are unable to gain access to the market price data for each block, we calculate the value of $\max(\Delta^\mathbb{B}_\mathbf{A})$ for each minute as $\max(\Delta^M_\mathbf{A})$ instead of each block.
Thus, the value of $\max(\Delta^M_\mathbf{A})$ at the $M^{th}$ minute from September 1st 2020, calculates from Equation~\ref{eq-max-delta} as 
\begin{equation}\label{eq-delta-minute}    
\max(\Delta^M_\mathbf{A}) = \max(d_M - \mathsf{A}.p) = d_M - \min(\mathsf{A}.p),
\end{equation}
where $d_M \in D_\mathbf{A}$ and $M \in [1,\cdots,N]$.

\subsubsection{Assumption} 

To calculate the value of $\max(\Delta^M_\mathbf{A})$ in Equation~\ref{eq-delta-minute} for $N$ data points, we assume that $PLF$ updates the state $\mathsf{A}$ at least once every $T \in \mathbb{Z}^*$ minutes for each asset $\mathbf{A}$.
Using this assumption, we calculate the value of $\max(\Delta^M_\mathbf{A})$ (from Rule~\ref{rule-update-A}) at the $M^{th}$ minute over the last $T$ minute(s) as 
\begin{equation}\label{eq-max-delta-T}
\max_T(\Delta^M_\mathbf{A}) = d_M - \min_i(d_{i}) \cdot \epsilon \ \forall i \in [M-T,\cdots,M].
\end{equation}

\begin{figure}[t]
    \centering
    \includegraphics[width=0.47\textwidth]{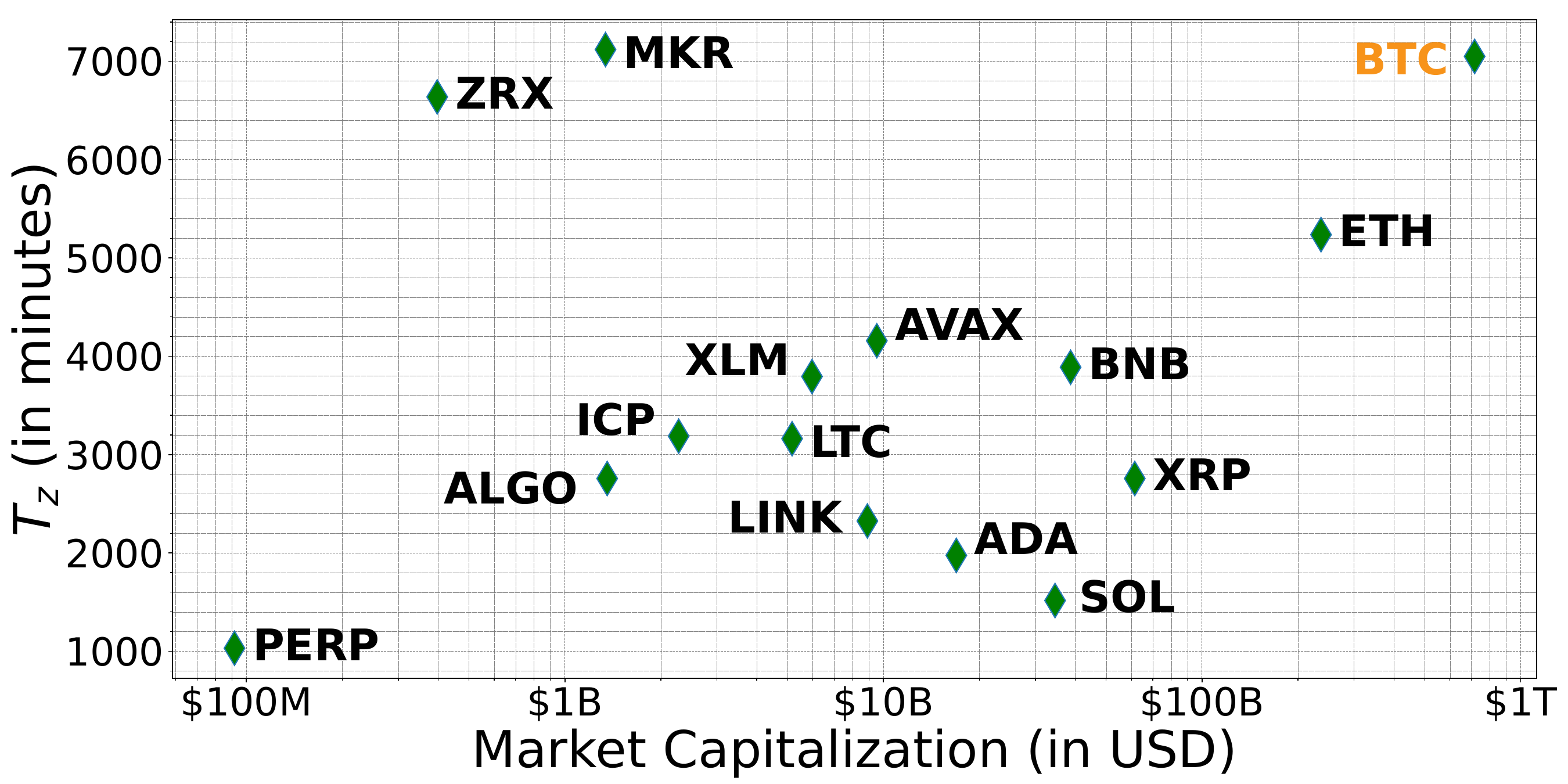}
    \caption{$T_z$ (in minutes) vs Market Capitalization (in USD) for 15 crypto-assets. The confidence value $z = 1 - 10^{-5}$. The safe collateralization ratio $\epsilon = 1.25$.}
    \label{fig: T_MAX-vs-coins}
\end{figure}

\subsubsection{Empirical analysis}

Moving forward, we aim to show that based on this assumption and three years of market data, the Confidence case (Case~\ref{case-confidence}) is satisfied for a reasonable value of $T$ for each dataset $D_\mathbf{A}$. 
Hence, we replace $\max(\Delta^\mathbb{B}_\mathbf{A})$ with $\max_T(\Delta^M_\mathbf{A})$ from Equation~\ref{eq-max-delta-T} to get Case~\ref{case-confidence} as $$\mathbb{P}(\max_T(\Delta^M_\mathbf{A}) \le 0) \ge z.$$

The aforementioned assumption implies that the lower the value of $T$, the higher the resources needed by $PLF$ to ensure security.
Hence, to analyze this parameter, we calculate the maximum value of $T$ such that the Confidence case (Case~\ref{case-confidence}) is satisfied as $T_z$. Here $T_z$ calculates as
\begin{equation}\label{eq-t-max}
T_z = \max(T)\ \text{if}\ \mathbb{P}(\max_T(\Delta^M_\mathbf{A}) \le 0) \ge z.
\end{equation}
Intuitively, it implies that assuming a $PLF$ uses $\mathbb{O}_\mathbf{A}$ as input at least once every $T_z$ minute(s), the Price-discrepancy case does not occur with a high probability of $z$, i.e. the Confidence case is satisfied.

Figure~\ref{fig: T_MAX-vs-coins} shows a plot of $T_z$ (Y-axis) vs Market Capitalization (X-axis) for each asset $\mathbf{A}$ (from Equation~\ref{eq-t-max}), where $\epsilon = 1.25 \And z = 1 - 10^{-5}$.
Here, we take $\epsilon = 1.25$ as it is a typical lower bound used for the safe collateralization ratio by a standard $PLF$ \cite{aave-docs}. 
This value of $\epsilon$ represents the worst-case scenario for $PLF$ as the lower the value of $\epsilon$, the higher is the value of $\max(\Delta^M_\mathbf{A})$, implying more arbitrage risk.
The figure illustrates that for each asset $\mathbf{A}$, the maximum value of $T$ such that the Confidence case (Case~\ref{case-confidence}) is satisfied is well over $1,000$ minutes, which is approximately 16 hours i.e. $T_z > 1,000$.
In the next sub-section, we show that this maximum value of $T$, which is at least $1,000$ minutes ($T_z > 1,000$), is more than the practical upper bound of $T$ based on factual evidence.

\subsection{Practicality Analysis}\label{sec: practical}

In this sub-section, we justify the practicality of \SecPLF\ by justifying its \textit{(i)} adaptability in case of price shocks in the market; \textit{(ii)} re-configurability in terms of $\epsilon$ and $z$; and \textit{(iii)} scalability in terms of algorithm complexities, and ease-of-implementation and integration with a standard $PLF$.

\subsubsection{Adaptability of \SecPLF\ in case of price shocks}

\begin{figure}[!t]
    \centering
    \includegraphics[width=0.48\textwidth]{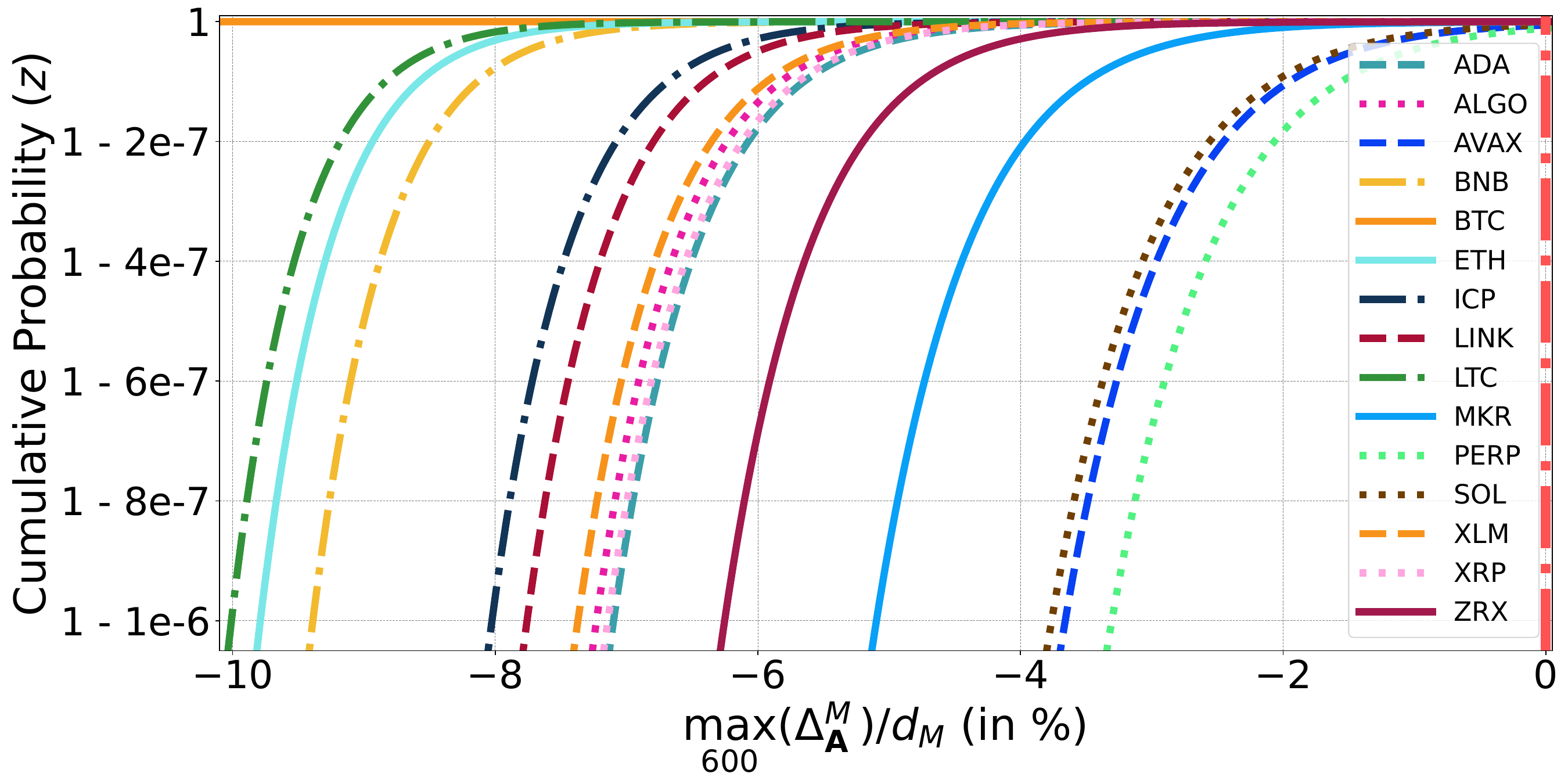}
    \caption{Cumulative Probability of $\max_{600}(\Delta^M_\mathbf{A})/d_M$ for each data set $D_\mathbf{A} = (d_1, \cdots, d_{N})$, where $N \approx 1.57 \cdot 10^6$. The safe collateralization ratio $\epsilon = 1.25$.}
    \label{fig: CPvsDelta600}
\end{figure}

To justify the adaptability of \SecPLF, we aim to justify our assumption, that a $PLF$ uses an oracle price as an input for each asset $\mathbf{A}$ at least once every $T$ minute(s).
Since $T_z$ is the maximum value of $T$ in \SecPLF\ such that the Confidence case is satisfied (i.e. Equation~\ref{eq-t-max}), the value of $T$ in our assumption must be less than or equal to $T_z$, i.e. $T \le T_z$.

Observe that for each asset in Figure~\ref{fig: T_MAX-vs-coins}, the value of $T_z$ is over $1,000$ minutes, i.e. $T_z > 1,000$. 
Therefore, to justify our assumption, we must show that $T \le 1,000$ in the real-world setting.
Since $PLF$s use these price oracles to monitor safe collateralization \cite{aave-docs, compound-docs}, a standard $PLF$ must use the oracle price of each asset at least once every 10 hours. Thus, implying that $T \le 600$ (in minutes), which justifies our assumption.

Further, figure~\ref{fig: CPvsDelta600} shows a plot of the Cumulative Density Function (CDF) of $\max_{600}(\Delta^M_\mathbf{A}) / d_M$ for each data set $D_\mathbf{A} = (d_1, \cdots, d_{N})$, where $N \approx 1.57 \cdot 10^6$. It illustrates that $\max_{t}(\Delta^M_\mathbf{A}) \le 0 \ \forall t \in [1,\cdots,600]$ with a high probability of at least $1 - 2 \cdot 10^{-7}$ for each dataset $D_\mathbf{A}$.
This probability is approximately two factors more than the $z$ value in the Confidence case (Case~\ref{case-confidence}), i.e. $$\mathbb{P}(\max_{t}(\Delta^M_\mathbf{A}) \le 0) \ge 1 - 2 \cdot 10^{-7} \gg z \ \forall t \in [1,\cdots,600]\  \forall D_\mathbf{A},$$ where $z = 1 - 10^{-5}$ and $\epsilon = 1.25$.

Hence, the Confidence case is justified with a reasonable margin even in the worst-case scenario for $PLF$.
This justification shows that, along with being a practical solution to tackle oracle manipulation attacks, \SecPLF\ is adaptable to the arbitrage risk that may arise in the Price-discrepancy case due to price shocks.

\subsubsection{Re-configurability of \SecPLF}

The re-configurability of SecPLF is established through the incorporation of two crucial parameters, i.e. $\epsilon$ and $z$. 
These two parameters are designed to quantify distinct dimensions of risk for a standard $PLF$ (Section~\ref{sec: PLF}). 

The safe collateralization ratio $\epsilon$ serves as a metric to quantify under-collateralization risk, providing a comprehensive assessment of the protocol's exposure to potential financial vulnerabilities arising from under-collateralization. 
On the other hand, the confidence value $z$ quantifies the probability of arbitrage risk, addressing the susceptibility of the protocol to arbitrage opportunities that may arise due to price shocks in the market. 

Integrating these two parameters in \SecPLF, provides a versatile and flexible solution for $PLF$s, enabling them to dynamically configure and fine-tune these risk parameters in response to evolving market conditions. 
The re-configurable nature of \SecPLF\ thus positions it as a robust and responsive solution capable of mitigating the aforementioned two risk factors, enhancing the overall security and stability of the DeFi landscape.

\subsubsection{Scalability of \SecPLF}

We justify \SecPLF\ as a scalable solution for standard $PLF$s (as modeled in Section~\ref{sec: PLF}), underpinned by careful consideration of key implementation details.
The algorithmic complexities of \SecPLF\ exhibit linear storage complexity, contingent upon the number of assets on $PLF$. 
Further, the \SecPLF\ algorithm runs in a constant time complexity.
This design ensures that the protocol can efficiently accommodate a growing number of assets without sacrificing computational efficiency. 

Similar to the TWAP oracles in \cite{mackinga2022twap} and the Timelock mechanisms in \cite{ezzat2022timelock}, \SecPLF\ algorithm seamlessly integrates as an independent layer solution in-between the oracle receive function and the oracle usage function of a standard $PLF$ as modeled in Section~\ref{sec: PLF}. 
This integration method minimizes disruption to existing $PLF$s, allowing for a straightforward adoption process. 

Moreover, the overhead cost of implementing \SecPLF\ can be evaluated with a focus on the storage of additional $n$ states in the smart contract for $n$ assets on $PLF$.
It is noteworthy that this overhead cost is minimal relative to the existing operational costs of $PLF$s \cite{link-analysis, SANKA2021103232}.
Thus, reinforcing the scalability of \SecPLF\ as an economically viable and operationally efficient solution for $PLF$s to tackle flash loan-driven oracle manipulation attacks.

\subsection{Comparative Analysis}

In this sub-section, we provide a comparative analysis of \SecPLF\ among the existing solutions introduced in Section~\ref{sec: related}.

\subsubsection{TWAP oracles and timelock mechanisms}
In contrast to existing solutions leveraging TWAP oracles (\cite{adams2022uniswap}) and timelock mechanisms within PLFs \cite{ezzat2022timelock}, \SecPLF\ introduces a novel solution that addresses the vulnerabilities associated with these approaches. 
The susceptibility of TWAP oracles to oracle manipulation attacks, as demonstrated in \cite{mackinga2022twap}, underscores their inherent security risks.
Additionally, Timelock mechanisms remain susceptible to arbitrage attacks in case of price shocks as they introduce a delay in oracle price usage \cite{ezzat2022timelock}.

\SecPLF\ strategically mitigates these vulnerabilities by providing robust security against oracle manipulation attacks (Theorem~\ref{theorem1}), and adaptability to price shocks in the market (Case~\ref{case-confidence}). 
The re-configurable nature of \SecPLF, quantified through parameters $\epsilon$ and $z$, positions it as a superior alternative capable of handling dynamic market conditions with heightened security.
Moreover, \SecPLF\ can be easily integrated with any standard $PLF$ as modeled in Section~\ref{sec: PLF}, further justifying its applicability.

\subsubsection{Staking and reputation systems}

Furthering the comparative analysis, \SecPLF\ betters itself from existing solutions relying on staking and reputation systems.
While staking and reputation systems are vulnerable to Sybil attacks, which introduces the risk of manipulation as evidenced in \cite{pauwels2022zkkyc}; \SecPLF\ stands resilient against oracle manipulation attacks, as affirmed by the \SecPLF\ Theorem.
This fundamental security feature positions \SecPLF\ as a trustworthy solution that does not compromise integrity even in the face of adversaries attempting to subvert the system.

\subsubsection{Multiple oracle sources}

In fortifying PLFs against oracle manipulation attacks, \SecPLF\ and the approach of using multiple oracle sources represent two distinctive strategies. 
\SecPLF\ relies on a formal security theorem, offering a theoretical foundation for its resilience to flash loan-driven oracle manipulation attacks. 
Conversely, the usage of multiple oracle sources adopts a pragmatic approach, which introduces weighted averages of price data from multiple oracles. 
This strategy seeks to thwart manipulation by requiring consensus or majority agreement among participating oracles. 
However, reliance on weighted averages of multiple oracle sources may be vulnerable to Sybil attacks, hence undermining its practicality. 
While \SecPLF\ emphasizes theoretical security, the selection between these approaches hinges on the specific needs and risk preferences of a PLF implementation.

Moreover, \SecPLF\ leverages the strengths of the solutions based on multiple oracle sources.
As evidenced in \cite{synthetix, hackmd-oracle}, multiple oracle usage improves the reliability and redundancy of $PLF$, enhancing overall security. 
\SecPLF\ prioritizes the oracle-agnostic property, that is the usage of an independent oracle price as an input to generate a safe-to-use price as an output. 
This oracle-agnostic property allows \SecPLF\ to leverage the strategy of multiple oracle sources, further justifying its practicality.

\subsubsection{Circuit breakers}

In the landscape of risk management, \SecPLF\ and circuit breakers in traditional finance \cite{qin2021cefi, Santoni1993CircuitBA, li2021impacts, Lauterbach1993StockMC} employ distinct strategies to tackle market volatility. 
Traditional circuit breakers, prevalent in established financial markets, temporarily halt trading during extreme price fluctuations to mitigate panic-driven behavior.
In contrast, \SecPLF\ is a re-configurable and proactive safe-guarding solution in decentralized finance (DeFi) offering zero downtime.
Thus, along with being resistant to oracle manipulation attacks, \SecPLF\ maintains the availability of $PLF$ in case of price shocks, unlike circuit breakers.

Further, \SecPLF's decentralized nature differs markedly from the centralized control of traditional circuit breakers, aligning with DeFi's core principles.
Its reconfigurability allows dynamic adjustments to primary risk parameters $z$ and $\epsilon$. 
Thus, \SecPLF\ allows $PLF$ to monitor itself based on the desired level of arbitrage ($z$) and under-collateralization ($\epsilon$) risks, a unique feature absent in circuit breakers in centralized finance.

\section{Conclusion}\label{sec:concl}

This paper set out to tackle oracle manipulation attacks on PLFs, which have been exacerbated with the advent of flash loans in the recent DeFi landscape. 
Through an in-depth analysis of the attack mechanism, formalizing both operational and adversary models for PLFs, we identify specific vulnerabilities that could be exploited.
This understanding led to the development of \SecPLF, a robust and efficient solution specifically tailored to counteract these attacks.

At the core of \SecPLF\ is the concept of tracking a price state for each crypto-asset and introducing specific price constraints to output a safe-to-use price, effectively rendering the oracle manipulation attack impossible to execute. 
This methodology ensures that a PLF only engages a price oracle if the last recorded price is within a defined threshold, effectively making potential attacks unprofitable.
The attributes of \SecPLF\ were discussed, emphasizing its proactive protection against attacks, adaptability in case of price shocks, re-configurability, and scalability.
Further, \SecPLF\ provides a holistic and secure approach to handling oracle price data, surpassing the limitations of existing solutions.

In conclusion, \SecPLF\ provides a robust, adaptable, and cost-effective solution against flash loan-driven oracle manipulation attacks. 
Addressing one of the significant vulnerabilities in the DeFi landscape helps enhance the security and reliability of DeFi platforms, thereby contributing to the further growth and success of decentralized finance.

\bibliographystyle{ACM-Reference-Format}
\bibliography{references}

% \appendix

\end{document}